\documentclass[11pt]{amsart}
\usepackage{times,amssymb,amsmath,epsfig,nicefrac,euscript,cite,mathrsfs,mathrsfs,amsaddr}

\newcommand\nc\newcommand
\nc\bfa{{\boldsymbol a}}\nc\bfA{{\boldsymbol A}}\nc\cA{{\mathcal A}}
\nc\bfb{{\boldsymbol b}}\nc\bfB{{\boldsymbol B}}\nc\cB{{\mathcal B}}
\nc\bfc{{\boldsymbol c}}\nc\bfC{{\boldsymbol C}}\nc\cC{{\mathcal C}}
\nc\sC{{\mathscr C}}
\nc\bfd{{\boldsymbol d}}\nc\bfD{{\boldsymbol D}}\nc\cD{{\mathcal D}}
\nc\bfe{{\boldsymbol e}}\nc\bfE{{\boldsymbol E}}\nc\cE{{\mathcal E}}
\nc\bff{{\boldsymbol f}}\nc\bfF{{\boldsymbol F}}\nc\cF{{\mathcal F}}
\nc\bfg{{\boldsymbol g}}\nc\bfG{{\boldsymbol G}}\nc\cG{{\mathcal G}}
\nc\bfh{{\boldsymbol h}}\nc\bfH{{\boldsymbol H}}\nc\cH{{\mathcal H}}
\nc\bfi{{\boldsymbol i}}\nc\bfI{{\boldsymbol I}}\nc\cI{{\mathcal I}}
\nc\bfj{{\boldsymbol j}}\nc\bfJ{{\boldsymbol J}}\nc\cJ{{\mathcal J}}
\nc\bfk{{\boldsymbol k}}\nc\bfK{{\boldsymbol K}}\nc\cK{{\mathcal K}}
\nc\bfl{{\boldsymbol l}}\nc\bfL{{\boldsymbol L}}\nc\cL{{\mathcal L}}
\nc\bfm{{\boldsymbol m}}\nc\bfM{{\boldsymbol M}}\nc\sM{{\mathscr M}}
\nc\bfn{{\boldsymbol n}}\nc\bfN{{\boldsymbol N}}\nc\cN{{\mathcal N}}
\nc\bfo{{\boldsymbol o}}\nc\bfO{{\boldsymbol O}}\nc\cO{{\mathcal O}}
\nc\bfp{{\boldsymbol p}}\nc\bfP{{\boldsymbol P}}\nc\cP{{\mathcal P}}
\nc\bfq{{\boldsymbol q}}\nc\bfQ{{\boldsymbol Q}}\nc\cQ{{\mathcal Q}}
\nc\bfr{{\boldsymbol r}}\nc\bfR{{\boldsymbol R}}\nc\cR{{\mathcal R}}
\nc\bfs{{\boldsymbol s}}\nc\bfS{{\boldsymbol S}}\nc\cS{{\mathcal S}}
\nc\bft{{\boldsymbol t}}\nc\bfT{{\boldsymbol T}}\nc\cT{{\mathcal T}}
\nc\bfu{{\boldsymbol u}}\nc\bfU{{\boldsymbol U}}\nc\cU{{\mathcal U}}
\nc\bfv{{\boldsymbol v}}\nc\bfV{{\boldsymbol V}}\nc\cV{{\mathcal V}}
\nc\bfw{{\boldsymbol w}}\nc\bfW{{\boldsymbol W}}\nc\cW{{\mathcal W}}
\nc\bfx{{\boldsymbol x}}\nc\bfX{{\boldsymbol X}}\nc\cX{{\mathcal X}}
\nc\bfy{{\boldsymbol y}}\nc\bfY{{\boldsymbol Y}}\nc\cY{{\mathcal Y}}
\nc\bfz{{\boldsymbol z}}\nc\bfZ{{\boldsymbol Z}}\nc\cZ{{\mathcal Z}}

\def\supp{\qopname\relax{no}{supp}}

\def\h_q{\qopname\relax{no}{h_q}}
\def\dist{d_H}
\def\avg{{\mathbb E}}

\newtheorem{theorem}{Theorem}
\newtheorem{definition}{Definition}
\newtheorem{lemma}[theorem]{Lemma}
\newtheorem{proposition}[theorem]{Proposition}

\newcommand{\Namea}{Arya Mazumdar}
\newcommand{\Addra}{Department of EECS / Research Laboratory of Electronics\\
Massachusetts Institute of Technology\\
Cambridge, MA 02139\\ email: aryam@mit.edu}
\allowdisplaybreaks

\title[Construction of Group Testing Scheme]{Construction of Almost Disjunct Matrices for Group Testing}
\author[A. Mazumdar]{\textsc{\Namea}}
\address{\textsc{\Addra}}
\begin{document}
\thanks{This work was supported in part by the US Air Force Office of Scientific Research under Grant No.  FA9550-11-1-0183, and by the National Science Foundation under Grant No.  CCF-1017772.}
\maketitle
\begin{abstract}
In a \emph{group testing} scheme, a set of tests is designed
to identify a small number $t$ of defective items among a large set (of size $N$) of items.
In the non-adaptive scenario the set of tests has to be designed
in one-shot. In this setting, designing a testing scheme is equivalent
to the construction of  a \emph{disjunct matrix}, an $M \times N$
matrix where the union of supports of any $t$ columns does not contain the support of any other column.
In principle, one wants to have such a matrix with minimum possible number $M$ of rows (tests).
One of the main ways of constructing disjunct matrices relies on \emph{constant weight
error-correcting codes} and their \emph{minimum distance}.
In this paper, we consider a relaxed definition of a disjunct matrix known as \emph{almost disjunct matrix}.
This concept is also studied under the name of \emph{weakly separated design} in the literature. 
 The relaxed definition
allows one to come up with group testing schemes where a close-to-one fraction of all
possible sets of defective items are identifiable. 
 Our main contribution is twofold. First, we go beyond the minimum distance analysis and 
 connect the \emph{average distance} of a constant
 weight code to the parameters of an almost disjunct matrix constructed from it.  
Our second contribution is
to  explicitly construct almost disjunct matrices based
on our average distance analysis,
that have much smaller number of rows than any previous explicit construction of disjunct matrices.
The parameters of our construction can be varied to cover a large range of relations for $t$ and $N$.
As an example of parameters, consider any absolute constant $\epsilon >0$ and 
$t$ proportional to $N^\delta, \delta >0$. With our method
it is possible to explicitly construct a group testing scheme
that identifies $(1-\epsilon)$ proportion of all possible defective  sets of size $t$ using only
$O\Big(t^{3/2}\sqrt{ \log(N/\epsilon)}\Big)$ tests. 
On the other hand, to form an explicit non-adaptive group testing scheme that works for
all possible defective sets of size $t$,
one requires $O(t^2 \log N)$ tests.
  \end{abstract}
  \newpage
\section{Introduction}
Combinatorial  group testing is an old and well-studied problem.
In the most general form it is assumed that there is a set of $N$ elements
among which at most $t$ are {\em defective}, i.e., special.
This set of defective items is called the {\em defective set} or {\em configuration.}
To find the defective set, one might
test all the elements individually  for {\em defects}, requiring $N$ tests. Intuitively, that would be a waste of resource if
$t \ll N$. On the other hand, to identify the defective configuration it is required to ask at least
$\log \sum_{i=0}^t \binom{N}{i} \approx t \log \frac{N}{t}$  yes-no questions. The main objective is to identify the defective
configuration with a number of tests that is as close to this minimum as possible.

In the group testing problem, a {\em group} of elements are tested together  and
if this particular group contains any defective
element the test result is positive. Based on the test results of this kind one {\em identifies} (with an efficient
algorithm) the  defective set with minimum possible number of tests.
The schemes (grouping of elements) can be adaptive, where
the design of one test may depend on the results of preceding tests. For a comprehensive survey
of adaptive group testing schemes we refer the reader to \cite{DH2000}.

In this paper we are interested in non-adaptive group testing schemes: here
all the tests are designed together. If the number of designed tests is $M$, then a
non-adaptive group testing scheme is equivalent to the design of a so-called binary {\em test matrix}
of size $M \times N$ where the $(i,j)$th entry is $1$ if the $i$th test includes the $j$th element;
it is $0$ otherwise. As the test results,  we see  the Boolean OR of the
columns corresponding to the defective entries.

Extensive research has been performed to find out the minimum
number of required tests $M$ in terms of the number of elements $N$ and the maximum number of
defective elements $t$. The best known lower bound says that it is necessary
to have $M = O(\frac{t^2}{\log{t}} \log N)$ tests \cite{DR1982, DRR1989}.
The existence of non-adaptive group testing
 schemes with $M = O(t^2 \log N)$ is also known for quite some time \cite{DH2000, HS1987}.
 
 Evidently, there is a gap by the factor of  $O(\log t)$ in these upper and lower bounds.
It is generally believed that it is hard to close the gap. On the other hand, for the adaptive
setting, schemes have been constructed  with as small as $O(t\log n)$ tests,  optimal
up to a constant factor \cite{DH2000, H1972}.

A construction of group testing schemes from error-correcting code matrices and using code concatenation
appeared in the seminal paper by Kautz and Singleton \cite{KS1964}.
Code concatenation is a way to construct binary codes from codes over a larger alphabet \cite{MS1977}.
In \cite{KS1964}, the authors
concatenate a $q$-ary ($q>2$) Reed-Solomon code with a unit weight code to use the resulting codewords as the columns of the testing matrix.
Recently in \cite{PR2008}, an explicit construction of a scheme with $M = O(t^2 \log N)$ tests
is provided. The construction of \cite{PR2008} is based on the idea of \cite{KS1964}: instead of
the Reed-Solomon code, they take a low-rate code that achieves the Gilbert-Varshamov bound
of coding theory \cite{MS1977, R2006}. Papers, such as \cite{DRM2000,Y1998}, also consider construction
of non-adaptive group testing schemes.

In this paper we  explicitly construct a non-adaptive scheme that requires a number of
test proportional to $t^{3/2}$.
However, we needed to relax the requirement of identifications of defective elements in a way
that makes it amenable for our analysis. 
This relaxed requirement schemes were considered under the name of {\em weakly separated designs} in
\cite{M1978} and \cite{Z2003}. Our definition of this relaxation appeared previously in the paper \cite{MRY2004}.
We (and \cite{M1978,Z2003,MRY2004}) aim for a scheme that successfully identifies a large fraction of 
all possible defective configurations.
Non-adaptive group testing has found applications
in multiple different areas, such as, multi-user communication \cite{BMTW1984, W1985}, DNA screening \cite{ND2000},
pattern finding \cite{MP2004} etc.
It can be observed that in many of these applications it would have been still useful
to have a scheme that identifies almost all different defective configurations if not
all possible defective configurations. 
It is known (see, \cite{Z2003}) that with this relaxation
it might be possible to reduce the number of tests to be proportional to $t\log N$. However
this result is not constructive.
The above relaxation and weakly separated designs form a parallel of similar works in 
 compressive sensing (see, \cite{CHJ2010,MB2011}) where recovery of almost all sparse signals from a
 generic random model is considered.
In the literature, other relaxed versions of the group testing problem have been studied as well.
For example, in \cite{GIS2008} it is assumed that recovering a large fraction of
defective elements is sufficient. There is also effort to form an information-theoretic
model for the group testing problem where test results can be noisy \cite{AS2010}.
In other versions of the group testing problem, a test may carry more than
one bit of information \cite{H1984, BKS1971}, or the test results are threshold-based (see \cite{C2010}
and references therein). Algorithmic aspects of the recovery schemes have
been studied in several papers. For example, papers \cite{INR2010} and \cite{NPR2011} provide
very efficient recovery algorithms for non-adaptive group testing.
\subsection{Results}
The constructions of \cite{KS1964, PR2008} and many others are based on  so-called
{\em constant weight error-correcting codes}, a set of binary vectors of same Hamming weight (number of ones).
The group-testing recovery property relies on the pairwise {\em minimum distance}
between the vectors of the code \cite{KS1964}. In this work,  we go beyond this
minimum distance analysis and relate the group-testing parameters to the
{\em average distance} of the constant weight code. This allows us to connect
weakly separated designs to error-correcting codes in a general way. Previously the connection 
between distances of the code and weakly separated designs
 was only known for the very specific family of {\em maximum distance
separable} codes \cite{MRY2004}, where much more information than the
average distance is evident.

Based on the newfound connection, we construct an explicit (constructible deterministically in polynomial time)
  scheme of non-adaptive group testing that can
identify all except an $\epsilon>0$ fraction of all defective sets of
size at most $t.$ To be specific,
we show that it is possible to explicitly construct a group testing scheme
that identifies $(1-\epsilon)$ proportion of all possible defective  sets of size $t$ using only
$8e t^{3/2} \log N \frac{\sqrt{\log\frac{2(N-t)}{\epsilon}}}{\log t- \log \log \frac{2(N-t)}{\epsilon}}$ tests for any $\epsilon >2(N-t)e^{-t}.$
It can  be seen that, with the relaxation in requirement, the  number of tests is
brought down to be proportional to $t^{3/2}$ from $t^2.$
This  allows us to
operate with a number of tests that was previously not possible in explicit constructions of non-adaptive group testing.
For a large range of values of $t$, namely $t$ being proportional to any positive power of $N$, i.e., $t\sim N^\delta$,
 and constant $\epsilon$
our scheme has number of tests only  about $\frac{8e}{\delta}t^{3/2} \sqrt{\log (N/\epsilon)})$.
 Our construction technique
is same as the scheme of \cite{KS1964,PR2008}, however with a finer analysis
relying on the distance properties of a linear code
we are able to achieve more.

In Section \ref{sec:disjunct}, we provide the necessary definitions and state one of the main results:
we state the connection between the parameters of a weakly separated design and the average distance of a 
constant weight code.
In Section \ref{sec:construction}
we discuss our  construction scheme. The proofs of our claims
can be found in Sections \ref{sec:proof1} and \ref{sec:construction}.

\section{Disjunct Matrices}\label{sec:disjunct}
\subsection{Lower bounds}
It is easy to see that, if an $M \times N$ binary  matrix gives a non-adaptive group testing scheme
that identify up to $t$ defective elements, then,
$\sum_{i=0}^t \binom{N}{i} \le 2^M.$
This means that for any group testing scheme,
\begin{equation}
\label{eq:bound1}
M \ge \log \sum_{i=0}^t \binom{N}{i}  \ge t \log\frac{N}{t}.
\end{equation}
Consider the case when one is interested in a scheme that identifies all possible except an $\epsilon$ fraction of
the different defective sets. Then it is required that,
\begin{equation}
\label{eq:bound2}
M \ge \log \Big((1-\epsilon) \binom{N}{t} \Big) \ge t \log\frac{N}{t}+\log (1-\epsilon).
\end{equation}
 Although
\eqref{eq:bound1} is proven to be a loose bound, it is shown in \cite{Z2003, M1978} that  \eqref{eq:bound2} is tight.

\subsection{Disjunct matrices}
The {\em support} of a vector $\bfx$ is the set of coordinates  where the vector has nonzero
entries. It is denoted by $\supp(\bfx)$. We use the usual set terminology, where a set $A$ contains $B$ if $B\subseteq A.$

\begin{definition}
An $M \times N$ binary matrix $A$ is called $t$-disjunct if the support of any column
is not contained in
the union of the supports of any other  $t$ columns.
\end{definition}
It is not very difficult to see that a $t$-disjunct matrix gives a group testing scheme
that identifies any defective set up to size $t$. On the other hand
any group testing scheme
that identifies any defective set up to size $t$ must be a $(t-1)$-disjunct
matrix \cite{DH2000}. To a great advantage,  disjunct matrices allow for a simple identification
algorithm that runs in time $O(Nt).$
Below we define  {\em relaxed} disjunct matrices. This definition appeared very closely in \cite{M1978,Z2003} 
and independently exactly in
\cite{MRY2004}.
\begin{definition}
For any $\epsilon >0$, an $M \times N$ matrix $A$ is called type-1 $(t,\epsilon)$-disjunct if  the set of  $t$-tuple of
columns (of size $\binom{N}t$) has a subset $\cB$ of size at least $(1-\epsilon)\binom{N}t$ with the following property:
for all $J \in \cB$, $\cup_{\kappa\in J}\supp(\kappa)$ does not contain support of any  column $\nu \notin J.$
\end{definition}
In other words, the union of supports of a
randomly and uniformly chosen set of  $t$ columns from a type-1 $(t,\epsilon)$-disjunct matrix does not contain the support of any
other column with probability at least $1-\epsilon$.
It is easy to see the following fact.
\begin{proposition}\label{prop:scheme}
A type-1 $(t,\epsilon)$-disjunct matrix gives a group testing scheme that can identify
all but at most a fraction $\epsilon >0$ of all possible defective configurations of size at most $t$.
\end{proposition}

The definition of disjunct matrix can be restated as follows: a matrix is $t$-disjunct if
any  $t+1$ columns  indexed by $i_1,\dots,i_{t+1}$ of the
matrix form a sub matrix which must have a
row that has exactly one $1$ in the  $i_j$th position and zeros in the
other positions, for $j =1, \dots, t+1.$ Recall that,
a \emph{permutation matrix} is a square binary $\{0,1\}$-matrix with exactly
one $1$ in each row and each column. Hence,
for a $t$-disjunct matrix, any $t+1$ columns form a sub-matrix that must
contain $t+1$ rows such that a $(t+1) \times (t+1)$
permutation matrix is formed of these rows and columns.
A statistical relaxation of  the above definition gives the following.
\begin{definition}
For any $\epsilon >0$, an $M \times N$ matrix $A$ is called type-2 $(t,\epsilon)$-disjunct if
the set of $(t+1)$-tuples of columns (of size $\binom{N}{t+1}$) has a subset
$\cB$ of size at least $(1-\epsilon)\binom{N}{t+1}$ with the following property: the $M\times(t+1)$
matrix formed by any element $J\in\cB$ must contain $t+1$ rows that form a
$(t+1)\times(t+1)$ permutation matrix.
\end{definition}
 In other words, with probability at least $1-\epsilon$, any
randomly and uniformly chosen $t+1$ columns from a type-2 $(t,\epsilon)$-disjunct
 matrix form a sub-matrix that must has $t+1$ rows such that a $(t+1) \times (t+1)$
permutation matrix can be formed.
It is clear that for $\epsilon =0$, the type-1 and type-2 $(t,\epsilon)$-disjunct matrices
are same (i.e., $t$-disjunct).
In the rest of the paper, we concentrate on the   design of  an $M \times N$ matrix $A$ that is
type-2 $(t,\epsilon)$-disjunct. Our technique can be easily extended to the construction of
type-1 disjunct matrices.

\subsection{Constant weight codes and disjunct matrices}
A binary $(M,N,d)$ code $\cC$ is a set of size $N$ consisting of $\{0,1\}$-vectors of length $M$.
 Here $d$ is the largest integer such that
any two vectors (codewords) of $\cC$ are at least
Hamming distance $d$ apart. $d$ is called the {\em minimum distance} (or {\em distance})
of $\cC.$ If all the codewords of $\cC$ have Hamming weight $w$,
then it is called a constant weight code. In that case we write $\cC$ is an
$(M,N,d,w)$ constant weight binary code.

Constant weight codes can give constructions of group testing schemes.
One just arranges the codewords as the columns of the test matrix. Kautz and Singleton
proved the following in \cite{KS1964}.
\begin{proposition}\label{prop:disj}
An $(M,N,d,w)$ constant weight binary code provides
a $t$-disjunct matrix where,
$
t = \Big\lfloor\frac{w-1}{w-d/2}\Big\rfloor.
$
\end{proposition}
\begin{proof}
The intersection of supports of any two columns has size at most $w -d/2$. Hence if
$w> t(w-d/2)$, support of any column will not be contained in the union of supports of any $t$ other columns.
\end{proof}

\subsection{$(t,\epsilon)$-disjunct matrices from constant weight codes}
We extend  Prop.~\ref{prop:disj} to have one of our main theorems. However, to do that we
need to define the {\em average distance} $D$ of a code $\cC$:
$$
D(\cC) =  \frac{1}{|\cC|-1} \min_{\bfx \in \cC}\sum_{\bfy \in \cC\setminus \{\bfx\}} \dist(\bfx,\bfy).
$$
Here $\dist(\bfx,\bfy)$ denotes the Hamming distance between $\bfx$ and $\bfy$.

\begin{theorem}\label{thm:main1}
Suppose, we have a constant weight binary code $\cC$ of size $N$, minimum distance $d$ and average distance $D$ such that every codeword has length $M$ and weight $w$. The test matrix obtained from the
code is type-2 $(t,\epsilon)$-disjunct for the largest $t$ such that,
$$
\alpha\sqrt{t\ln \frac{2(t+1)}{\epsilon}} \le \frac{w-1-t(w-D/2)}{w-d/2}
$$
holds. Here $\alpha$ is any absolute constant greater than or equal to $\sqrt{2}(1+t/(N-1))$.
\end{theorem}
The proof of this theorem is deferred until after the following remarks.

\vspace{0.1in}
\emph{Remark:} By a simple change in the proof of the Theorem \ref{thm:main1}, it is possible to
see that the test matrix is type-1 $(t,\epsilon)$-disjunct if,
$$
\alpha\sqrt{t\ln\frac{2(N-t)}{\epsilon}} \le \frac{w-1-t(w-D/2)}{w-d/2},
$$
for an absolute constant $\alpha.$

One can compare the results of Prop.~\ref{prop:disj} and Theorem~\ref{thm:main1}
to see the improvement achieved as we relax the definition of disjunct matrices.
This will lead to the final improvement on the parameters of Porat-Rothschild
construction \cite{PR2008}, as we will see in Section~\ref{sec:construction}.

\section{Proof of Theorem \ref{thm:main1}}\label{sec:proof1}
This section is dedicated to the proof of  Theorem \ref{thm:main1}.
Suppose, we have a constant weight binary code $\cC$ of size $N$ and  minimum distance $d$ such 
that every codeword has  length $M$ and weight $w$.  Let the average distance of the code be $D.$ 
Note that this code is fixed:
we will prove a property of this code by probabilistic method .

Let us now chose $(t+1)$ codewords randomly and uniformly from all possible $\binom{N}{t+1}$ choices.
Let the randomly chosen codewords  are $\{\bfc_1,\bfc_2,\dots, \bfc_{t+1}\}.$
In what follows, we adapt the proof of Prop.~\ref{prop:disj} in a probabilistic setting.

Define the
random variables for $i =1,\dots, t+1,$
$
Z^{i} =  \mathop{\sum_{j = 1}^{t+1}}_{j \ne i} \Big(w - \frac{\dist(\bfc_i, \bfc_j)}2\Big).
$
Clearly, $Z^i$ is the maximum possible size of the portion of the support of $\bfc_i$ that
is common to  at least one of $\bfc_j, j =1,\dots,t+1, j\ne i.$ Note that the size of support of $\bfc_i$
is $w$.
Hence, as we have seen in the proof of Prop.~\ref{prop:disj}, if $Z^i$ is
less than $w$ for all $i =1, \dots, t+1$,  then the $M\times (t+1)$ matrix formed
by the $t+1$ codewords must contain $t+1$ rows such that a $(t+1) \times (t+1)$
permutation matrix can be formed. Therefore, we aim to find the probability
$\Pr(\exists i \in \{1,\dots,t+1\} : Z^i \ge w)$  and show it to be bounded above
by $\epsilon$ under the condition of the theorem.

 As the variable $Z^i$s are identically distributed, we see that,
$$
\Pr(\exists i \in \{1,\dots,t+1\} : Z^i \ge w) \le (t+1) \Pr(Z^1 \ge w).
$$
In the following, we will find an upper bound on $ \Pr(Z^1 \ge w).$

Define,
$$
Z_i =\avg\Big( \sum_{j = 2}^{t+1} \Big(w - \frac{\dist(\bfc_1, \bfc_j)}2\Big) \mid  {\dist(\bfc_1, \bfc_k)}, k=2,3,\dots, i  \Big).
$$
Clearly,
$
Z_1 = \avg \Big( \sum_{j = 2}^{t+1} \Big(w - \frac{\dist(\bfc_1, \bfc_j)}2\Big) \Big),
$
and
$
Z_{t+1} =  \sum_{j = 2}^{t+1} \Big(w - \frac{\dist(\bfc_1, \bfc_j)}2\Big) = Z^1.
$

Now,
\begin{align*}
Z_1 &= \avg \Big( \sum_{j = 2}^{t+1} \Big(w - \frac{\dist(\bfc_1, \bfc_j)}2\Big) \Big)
= tw - \frac12 \avg \sum_{j = 2}^{t+1}\dist(\bfc_1, \bfc_j),
\end{align*}
where the expectation is over the randomly and uniformly chosen $(t+1)$ codewords
from all possible $\binom{N}{t+1}$ choices. Note,
\begin{align*}
&\avg \sum_{j = 2}^{t+1}\dist(\bfc_1, \bfc_j) =  \sum_{i_1<i_2<\dots<i_{t+1}}\frac1{\binom N{t+1}}\sum_{j=2}^{t+1} \dist(\bfc_{i_1}, \bfc_{i_j})\\
& =\frac1{(t+1)!\binom N{t+1}}\mathop{\sum_{i_l\ne i_m}}_{1\le l \ne m\le t+1} \sum_{j=2}^{t+1}
\dist(\bfc_{i_1}, \bfc_{i_j})=\frac1{N(N-1)}\sum_{j=2}^{t+1}\sum_{i_1=1}^N\sum_{i_j\ne i_1}
 \dist(\bfc_{i_1}, \bfc_{i_j})\\
&=\sum_{j=2}^{t+1}\avg \dist(\bfc_{i_1}, \bfc_{i_j}) \le tD,
\end{align*}
where the  expectation on the last but one line is over a uniformly chosen pair of distinct random codewords
of $\cC.$ Hence,
$$
Z_1  \le t(w-D/2).
$$

We start with the lemma below.
\begin{lemma}\label{lem:marti}
The sequence of random variables $Z_i, i =1, \dots , t+1,$ forms a martingale.
\end{lemma}
The statement is true by construction. For completeness we present a  proof that is deferred to Appendix \ref{app:one}.
Once we have proved that the sequence is a martingale, we show that it is a bounded-difference
martingale.
\begin{lemma}\label{lem:bounded}
For any $i = 2,\dots,t+1$,
$$
|Z_i -Z_{i-1}| \le (w-d/2) \Big(  1 +  \frac{t-i+1}{N-i}\Big).
$$
\end{lemma}
The proof is deferred to Appendix \ref{app:two}.

Now using Azuma's inequality for martingale with bounded difference \cite{Mc1989},
we have,
$$
\Pr\Big(|Z_{t+1} -Z_1|>  \nu) \le 2\exp\Big(-\frac{\nu^2}{2(w-d/2)^2\sum_{i=2}^{t+1}c_i^2} \Big),
$$
where, $c_i =   1 +  \frac{t-i+1}{N-i}$.
This implies,
$$
\Pr\Big(|Z_{t+1}| > \nu + t(w-D/2)\Big) \le 2\exp\Big(-\frac{\nu^2}{2(w-d/2)^2\sum_{i=2}^{t+1}c_i^2} \Big).
$$
Setting, $\nu = w -1-  t(w-D/2)$, we have,
\begin{align*}
\Pr\Big(Z^1 > w-1\Big)
&\le 2\exp\Big(-\frac{(w -1-  t(w-D/2))^2}{2(w-d/2)^2\sum_{i=2}^{t+1}c_i^2} \Big).
\end{align*}
Now,
\begin{align*}
\sum_{i=2}^{t+1}c_i^2 &\le t\Big(1+\frac{t-1}{N-2}\Big)^2.
\end{align*}
Hence,
\begin{align*}
\Pr(\exists i \in \{1,\dots,t+1\} : Z^i \ge w) &\le 2(t+1)\exp\Big(-\frac{(w -1-  t(w-D/2))^2}{2t(w-d/2)^2\Big(1+  \frac{t-1}{N-2}\Big)^2} \Big)< \epsilon,
\end{align*}
when,
$$
d/2 \ge w - \frac{w-1-t(w-D/2)}{\alpha\sqrt{t\ln\frac{2(t+1)}{\epsilon}}},
$$
and $\alpha$ is a constant greater than $\sqrt{2}\Big(1+  \frac{t-1}{N-2}\Big).$

\section{Construction}\label{sec:construction}
As we have seen in Section~\ref{sec:disjunct}, constant weight codes can be used to
produce disjunct matrices. Kautz and Singleton \cite{KS1964} gives a construction
of constant weight codes that results in good disjunct matrices.
In their construction, they start with a Reed-Solomon (RS) code, a
$q$-ary error-correcting code  of length
$q-1.$ For a detailed discussion of RS codes we refer the reader
to the standard textbooks of coding theory \cite{MS1977,R2006}.
Next they replace the $q$-ary symbols in the codewords by unit weight
binary vectors of length $q$. The mapping from $q$-ary symbols to length-$q$
unit weight binary vectors is bijective: i.e., it is $0\to 100\dots0; 1\to 010\dots0;\dots;
q-1\to 0\dots01.$ We refer to this mapping as $\phi.$ As a result, one obtains a set of  binary vectors
of length $q(q-1)$ and constant weight $q.$ The size of the resulting binary code
is same as the size of the RS code, and the distance of the binary code
is twice that of the distance of the RS code.

\subsection{Consequence of Theorem~\ref{thm:main1} in Kautz-Singleton construction}
For a $q$-ary RS code of size $N$ and length $q-1$, the
minimum distance is $q-1-\log_q{N}+1 = q-\log_q{N}.$ Hence,
the Kautz-Singleton construction is a constant-weight code with length $M=q(q-1)$, weight $w=q-1$,
size $N$ and distance $2(q-\log_q{N})$. Therefore,
from Prop.~\ref{prop:disj},
we have a $t$-disjunct matrix with,
$$
t = \frac{q-1-1}{q-1-q+\log_q{N}} = \frac{q-2}{\log_q{N}-1}\approx \frac{q\log q}{\log N} \approx \frac{\sqrt{M}\log M}{2\log N}.
$$
On the other hand, note that, the average distance of the RS code is
$\frac{N}{N-1}(q-1)(1-1/q).$ Hence the average distance of the resulting
constant weight code from Kautz-Singleton construction will be
$$
D= \frac{2N(q-1)^2}{q(N-1)}.
$$
Now, substituting these values in Theorem~\ref{thm:main1}, we have a type-1
$(t,\epsilon)$ disjunct matrix, where,
$$
\alpha\sqrt{t\ln\frac{2(N-t)}{\epsilon}} \le \frac{(q-t)\log q}{\log N}  \approx \frac{(\sqrt{M}-t)\log M}{2\log N}.
$$
Suppose $t \le \sqrt{M}/2$. Then,
$$
M (\ln M)^2 \ge 4\alpha^2 t (\ln N)^2 \ln\frac{2(N-t)}{\epsilon}
$$
This basically restricts $t$ to be about $O(\sqrt{M}).$ Hence, Theorem~\ref{thm:main1} does not
obtain any meaningful improvement from the Kautz-Singleton construction except in special cases.

There are two places where the Kautz-Singleton construction can be improved:
1) instead of Reed-Solomon code one can use any other
$q$-ary code of different length, and 2) instead of the mapping $\phi$
any binary constant weight code  of size $q$ might have been used.
For a general discussion we refer the reader to \cite[\S7.4]{DH2000}.
In the recent work \cite{PR2008}, the mapping $\phi$ is kept the same, while
the RS code has been changed to a $q$-ary code that achieve the Gilbert-Varshamov
bound \cite{MS1977, R2006}.

In our construction of disjunct matrices we follow the footsteps of \cite{KS1964, PR2008}.  However, we exploit some  property
of the resulting scheme (namely, the average distance) and do a
finer analysis that was absent from the previous works such as \cite{PR2008}.

\subsection{$q$-ary code construction}
We choose $q$ to be a power of a prime number and write $q =\beta t$, for some constant $\beta>2.$
The value of $\beta$ will be chosen later.
Next, we construct a {\em linear} $q$-ary code
of size $N$, length $M_q$ and minimum distance $d_q$ that
achieves the Gilbert-Varshamov bound \cite{MS1977,R2006}, i.e.,
\begin{equation}\label{eq:GV}
\frac{\log_q N}{M_q} \ge 1- \h_q\Big(\frac{d_q}{M_q}\Big)-o(1),
\end{equation}
where $\h_q$ is the $q$-ary entropy function defined by,
$$
\h_q(x) = x\log_q \frac{q-1}{x} +(1-x)\log_q\frac1{1-x}.
$$

Porat and Rothschild \cite{PR2008} show that it is
possible to construct  in time $O(M_q N)$ a $q$-ary code that achieves the Gilbert-Varshamov (GV)
bound. To have such construction, they exploit the
following  well-known fact: a $q$-ary linear code with random generator matrix
achieves the  GV bound with high probability \cite{R2006}. To have an explicit construction
of such codes, a derandomization method known as the method of conditional
expectation \cite{AS2000} is used. In this method, the entries of the generator
matrix of the code are chosen one-by-one so that the minimum distance of the resulting code
does not go below the value prescribed by \eqref{eq:GV}. For a detail description of
the procedure, see \cite{PR2008}.

With the above construction with proper parameters we can have a disjunct matrix with
the following property.
\begin{theorem}\label{thm:main2}
Suppose $\epsilon > 2(t+1)e^{-at}$ for some constant $a >1$. It is possible to explicitly construct a type-2 $(t,\epsilon)$-disjunct matrix of size $M\times N$ where
$$
M  = O\Big(t^{3/2} \ln N \frac{\sqrt{\ln\frac{2(t+1)}{\epsilon}}}{\ln t- \ln \ln \frac{2(t+1)}{\epsilon}+\ln(4a)} \Big).
$$
\end{theorem}
To prove this theorem we need the following identity implicit in \cite{PR2008}.
We present the proof here for completeness.
\begin{lemma}\label{lem:identity}
For any $q>s$,
$$
1- \h_q(1-1/s) = \frac{1}{s\ln q}\Big(\ln\frac{q}{s} +\frac{s}{q} -1\Big) - o\Big(\frac1{s\ln q}\Big).
$$
\end{lemma}
The proof of this is deferred to Appendix \ref{app:three}.
Now we are ready to prove Theorem~\ref{thm:main2}.

\begin{proof}[Proof of Theorem~\ref{thm:main2}]
We follow the Kautz-Singleton  code construction. 
We take a linear $q$-ary code $\cC'$ of length $M_q \triangleq \frac{M}{q}$, size $N$ and
minimum distance $d_q  \triangleq \frac{d}{2}.$ Each $q$-ary symbol in the codewords is then
replaced with a binary indicator vector of length $q$ (i.e., the binary vector whose all entries
are zero but one entry, which is $1$) according to the map $\phi$. As a result we have a binary code $\cC$ of length $M$ and
size $N$. The minimum distance of the code is $d$ and the codewords are of constant weight $w = M_q = \frac{M}{q}.$
The average distance of this code is twice the average distance of the $q$-ary code. As $\cC'$ is linear
(assuming it has no all-zero coordinate), it has average distance equal to
\begin{align*}
\frac{1}{N-1}\sum_{j =1}^{M_q} j A_j &= \frac{N}{N-1}\sum_{j=0}^{M_q}j\binom{M_q}{j}(1-1/q)^j(1/q)^{M_q-j}\\
& =  \frac{N}{N-1}M_q(1-1/q),
\end{align*}
where $A_j$ is the number of codewords of weight $j$ in $\cC'.$
Here we use the fact that the average of the  distance between any two randomly chosen codewords of
a nontrivial linear code is equal to that of a binomial random variable  \cite{MS1977}.
Hence the constant weight code $\cC$ has average distance
$$
D =  \frac{2N}{N-1}M_q(1-1/q).
$$
The resulting matrix will be $(t,\epsilon)$-disjunct if the condition of Theorem \ref{thm:main1} is satisfied,
i.e.,
\begin{align*}
d_q \ge M_q - \frac{M_q -1 -t(M_q -  \frac{N}{N-1}M_q(1-1/q))}{\alpha\sqrt{t\ln\frac{2(t+1)}{\epsilon}}}= M_q -  \frac{M_q -1 -\frac{tM_q}{N-1}(N/q-1)}{\alpha\sqrt{t\ln\frac{2(t+1)}{\epsilon}}}
\end{align*}
or if, $
d_q \ge M_q -  \frac{M_q -1 -\frac{tM_q}{q}}{\alpha\sqrt{t\ln\frac{2(t+1)}{\epsilon}}}.
$

To construct a desired $q$-ary code, we use the ideas of \cite{PR2008}
where the explicitly constructed codes  meet the Gilbert-Varshamov bound. It is possible to construct in
time polynomial in $N,M_q$, a $q$-ary code of length $M_q$, size $N$ and distance $d_q$ when
\begin{align*}
\frac{\log_q N}{M_q} \le 1- \h_q\Big(\frac{d_q}{M_q}\Big).
\end{align*}

Therefore, explicit polynomial time construction of a type-2 $(t,\epsilon)$-disjunct matrix will be
possible whenever,
\begin{align}\label{eq:cond}
\frac{\log_q N}{M_q} \le 1- \h_q\Big(1-\frac{1-1/M_q-\frac{t}{q}}{\alpha\sqrt{t\ln\frac{2(t+1)}{\epsilon}}}\Big).
\end{align}

Let us now use the fact that we have taken  $q= \beta t$ to be a prime power
  for some constant $\beta$. Let us chose $ \beta > 2e \alpha \sqrt{a}+1.$

Hence,
$$
\frac{1-1/M_q-\frac{t}{q}}{\alpha\sqrt{t\ln\frac{2(t+1)}{\epsilon}}} = \frac{1}{\gamma\sqrt{t\ln\frac{2(t+1)}{\epsilon}}} = \frac1s (\text{say}),
$$
for an absolute constant $\gamma \approx \frac{\alpha\beta}{\beta-1}.$
At this point, we see,
$$
\frac{q}{s} = \frac{\beta t}{\gamma\sqrt{t\ln\frac{2(t+1)}{\epsilon}}} > \frac{\beta}{\gamma \sqrt{a}} > 2e,
$$
from the condition on $\epsilon$ and the values of $\beta,\gamma.$
Now, using Lemma~\ref{lem:identity},
the right hand side of Eqn.~\eqref{eq:cond} equals to
\begin{align*}
   \frac{1}{s\ln q}\Big(\ln\frac{q}{s} +\frac{s}{q} -1\Big) - o\Big(\frac1{s\ln q}\Big) \ge \frac{\ln\frac{q}{s}-1 -o(1)}{s\ln q}.
\end{align*}
Then
  explicit polynomial time construction of a type-2  $(t,\epsilon)$-disjunct matrix will be
possible whenever,
\begin{align*}
\frac{\log_q N}{M_q} \le \frac{\ln \frac{\beta t }{\gamma\sqrt{t\ln\frac{2(t+1)}{\epsilon}}}-1-o(1)}{\gamma\sqrt{t\ln\frac{2(t+1)}{\epsilon}} \ln(\beta t)}
\end{align*}
or,
\begin{align*}
M = qM_q &\ge \beta t \frac{\ln N}{\ln(\beta t)}\frac{\gamma\sqrt{t\ln\frac{2(t+1)}{\epsilon}} \ln(\beta t)}{\ln \frac{\beta t }{\gamma\sqrt{t\ln\frac{2(t+1)}{\epsilon}}}-1-o(1)} \\
& = \beta\gamma t^{3/2} \ln N \frac{\sqrt{\ln\frac{2(t+1)}{\epsilon}}}{\frac12(\ln t- \ln \ln \frac{2(t+1)}{\epsilon})+\ln\frac\beta\gamma -1-o(1)}.
\end{align*}
The condition on $\epsilon$ and the value chosen for $\beta$ ensure that the denominator is strictly positive. Hence it suffices to have,
$$
M \ge 2\beta\gamma t^{3/2} \ln N \frac{\sqrt{\ln\frac{2(t+1)}{\epsilon}}}{\ln t- \ln \ln \frac{2(t+1)}{\epsilon}+\ln (4a) -o(1)}.
$$
\end{proof}
\vspace{0.1in}

Note that the implicit constant in Theorem~\ref{thm:main2} is proportional to $\sqrt{a}.$ We have not particularly tried to optimize the
constant. However even then the value of the constant is about $8e\sqrt{a}.$

\vspace{0.1in}
\emph{Remark:} As in the case of Theorem \ref{thm:main1}, with a simple change in the proof,
 it is easy to
see that one can construct a test matrix that is type-1 $(t,\epsilon)$-disjunct if,
$$
M  = O\Big(t^{3/2} \ln N \frac{\sqrt{\ln\frac{2(N-t)}{\epsilon}}}{\ln t- \ln \ln \frac{2(N-t)}{\epsilon}+\ln(4a)}\Big),
$$
for any $\epsilon > 2(N-t)e^{-at},$ and a constant $a$.

It is clear from Prop.~\ref{prop:scheme} that a type-1 $(t,\epsilon)$ disjunct
matrix is equivalent to a group testing scheme. Hence, as
 a consequence of Theorem~\ref{thm:main2} (specifically, the remark above),
we will be able to construct a testing scheme with $$O\Big(t^{3/2} \log N \frac{\sqrt{\log\frac{2(N-t)}{\epsilon}}}{\log t- \log \log \frac{2(N-t)}{\epsilon}}\Big)$$
tests. Whenever the defect-model is such that all the possible defective sets of
size $t$ are equally likely and there are no more than $t$ defective elements,
the above group testing scheme will be successful with probability at lease $1-\epsilon.$

Note that, if $t$ is proportional to any positive power of $N$, then $\log N$ and
$\log t$ are of same order. Hence it will be possible to have the above testing
scheme with $O(t^{3/2}\sqrt{\log (N/\epsilon)})$ tests, for any  $\epsilon> 2(N-t)e^{-t}$.

\section{Conclusion}
In this work we show that it is possible to construct non-adaptive 
group testing schemes with small number of tests that identify a
uniformly chosen random defective configuration with high probability.
To construct a $t$-disjunct matrix 
one starts with the simple relation between the minimum distance $d$ of a
constant $w$-weight code and $t.$ This is an example of a scenario where
a pairwise property (i.e., distance) of the elements of a set is
translated into a property of $t$-tuples.

Our method of analysis provides a general way to prove that a property
holds for almost all $t$-tuples of elements from a set based on the mean pairwise statistics
of the set. Our method will be useful in many areas of applied combinatorics,
such as digital fingerprinting or design of
key-distribution schemes, where such a translation is evident. For example, with our method
new results can be obtained for
 the cases of cover-free codes \cite{DVTM2002,SWJ2000,KS1964}, traceability and frameproof
codes \cite{CFN1994,SSW2001}. This is the subject of our 
ongoing work.




\appendix
\section{Proof of Lemma \ref{lem:marti}}\label{app:one}
We have   created a  sequence here that is a martingale by construction.  This is a standard
method due to Doob  \cite{D1953, Mc1989}. Let,
$$
w - \frac{\dist(\bfc_1, \bfc_j)}2 = Y_j.
$$
Consider the $\sigma$-algebras $\cF_k$, $k =0, \dots, t+1$, where
$\cF_0 = \{\emptyset, [N]\}$ and $\cF_k$ is generated by the partition of the set of $\binom N {t+1}$ possible
choices for $(t+1)$-sets
into $\binom N k$ subsets with the fixed value of the first
$k$ indices, $1\le k\le t+1.$ The sequence of increasingly refined partitions
$\cF_1\subset\cF_2\subset\dots\subset\cF_{t+1}$ forms a filtration such that
$Z_k$ is measurable with respect to $\cF_k$ (is constant on the atoms of the
partition).
 
We have,
\begin{align*}
Z_i &= \avg \Big( \sum_{j=2}^{t+1} Y_j \mid Y_2,\dots, Y_i \Big)\\
&=  \sum_{j=2}^{i} Y_j + \avg \Big( \sum_{j=i+1}^{t+1} Y_j \mid Y_2,\dots, Y_i \Big)\\
&=  Z_{i-1} +Y_i + \avg \Big( \sum_{j=i+1}^{t+1} Y_j \mid Y_2,\dots, Y_i \Big) - \avg \Big( \sum_{j=i}^{t+1} Y_j \mid Y_2,\dots, Y_{i-1} \Big).
\end{align*}
We then have,
\begin{align*}
\avg \Big(Z_i \mid Z_1,\dots,Z_{i-1}\Big)
&= Z_{i-1} + \avg\Big(Y_i \mid Z_1,\dots,Z_{i-1} \Big) \\
&+\avg\Big(\avg\Big( \sum_{j=i+1}^{t+1} Y_j \mid Y_2,\dots, Y_i\Big)\mid Z_1,\dots,Z_{i-1} \Big)\\
&-\avg\Big( \avg \Big( \sum_{j=i}^{t+1} Y_j \mid Y_2,\dots, Y_{i-1} \Big)\mid Z_1,\dots,Z_{i-1}\Big)\\
& = Z_{i-1} + \avg\Big(Y_i \mid Z_1,\dots,Z_{i-1} \Big) \\
&+ \avg\Big( \sum_{j=i+1}^{t+1} Y_j \mid Z_1,\dots,Z_{i-1}\Big)-\avg \Big( \sum_{j=i}^{t+1} Y_j \mid Z_1,\dots,Z_{i-1} \Big)\\
&= Z_{i-1}.
\end{align*}

\section{Proof of Lemma \ref{lem:bounded}}\label{app:two}
Let us again assume that, $$
w - \frac{\dist(\bfc_1, \bfc_j)}2 = Y_j.
$$
We have,

\begin{align*}
&|Z_i -Z_{i-1}| \\&=  \Big|\avg \Big( \sum_{j=2}^{t+1} Y_j \mid Y_2,\dots, Y_i \Big)  - \avg \Big( \sum_{j=2}^{t+1} Y_j \mid Y_2,\dots, Y_{i-1} \Big)\Big|\\
&\le \mathop{\max_{0\le a,b}}_{\le w-d/2}  \Big|\avg \Big( \sum_{j=2}^{t+1} Y_j \mid Y_2,\dots, Y_i=a \Big)  - \avg \Big( \sum_{j=2}^{t+1} Y_j \mid Y_2,\dots, Y_{i-1}, Y_i =b \Big)\Big|\\
&= \mathop{\max_{0\le a,b}}_{\le w-d/2}  \Big|\sum_{j=1}^{t+1}\Big(\avg \Big(  Y_j \mid Y_2,\dots, Y_i=a \Big)  - \avg \Big( Y_j \mid Y_2,\dots, Y_{i-1}, Y_i =b \Big)\Big)\Big|\\
&=  \mathop{\max_{0\le a,b}}_{\le w-d/2}  \Big| a-b + \sum_{j=i+1}^{t+1}\Big(\avg \Big(  Y_j \mid Y_2,\dots, Y_i=a \Big)
- \avg \Big( Y_j \mid Y_2,\dots, Y_i =b \Big)\Big)\Big|\\
&\le   \mathop{\max_{0\le a,b}}_{\le w-d/2}  \Big| w-d/2 +\sum_{j=i+1}^{t+1} \Big[\avg \Big(w-\frac{\dist(\bfc_1,\bfc_j)}2 \mid \dist(\bfc_1,\bfc_2),\dots, \dist(\bfc_1,\bfc_i)=2(w-a) \Big)  \\
&\hspace{0.7in}- \avg \Big( w-\frac{\dist(\bfc_1,\bfc_j)}2 \mid  \dist(\bfc_1,\bfc_2),\dots, \dist(\bfc_1,\bfc_i)=2(w-b) \Big)\Big]\Big|\\
&\le \Big| w-d/2 + \sum_{j=i+1}^{t+1} \frac{(w-d/2)}{N-1-(i-1)}\Big|\\
&= (w-d/2)  \Big(  1 +  \frac{t-i+1}{N-i}\Big)\\
&=(w-d/2)c_i,
\end{align*}

where $c_i =    1 +  \frac{t-i+1}{N-i}.$

\section{Proof of Lemma \ref{lem:identity}}\label{app:three}
\begin{proof}
The proof is straight-forward and uses the following approximation:
$$
\ln x - \ln(x-1) = \frac1x +o\Big(\frac1x\Big).
$$We have,
\begin{align*}
1- \h_q(1-1/s) & = \frac1{\ln q }\Big(1-\frac1s\Big)\Big((\ln q -\ln(q-1)) - (\ln s-\ln(s-1))\Big) \\
&+\frac1{s\ln q }\ln\frac{q}{s}\\
& = \frac1{\ln q }\Big(\frac1q -\frac1{qs} -\frac1s+\frac1{s^2} +\frac1s\ln\frac{q}{s}\Big) - o\Big(\frac1{\ln q}\Big(\frac1s-\frac1q\Big)\Big)\\
& = \frac1{s\ln q }\Big(\frac{s}{q} -1 + \ln\frac{q}{s}\Big)-o\Big(\frac1{s\ln q}\Big).
\end{align*}
\end{proof}

\end{document}